\documentclass[letterpaper]{article}

\usepackage[english]{babel}
\usepackage[T1]{fontenc}
\usepackage[utf8x]{inputenc}
\usepackage{fullpage}

\usepackage{amsthm}
\usepackage{mathrsfs}
\usepackage{amsfonts}
\usepackage{amssymb}
\usepackage{amsmath}
\usepackage{enumerate}

\usepackage{color}
\usepackage{graphicx}
\usepackage{graphics}
\usepackage{xspace}
\usepackage[ruled,vlined,algo2e]{algorithm2e}
\usepackage[color=green]{todonotes}
\usepackage{subcaption}

\usepackage{algorithmic}
\usepackage{algorithm}
\floatname{algorithm}{Procedure}

\def\cG{\mathcal{G}}

\newcommand{\contracted}[2]{\ensuremath{#1^{ch}(#2)}\xspace}
\newcommand{\contractedd}[2]{\ensuremath{#1^{str}(#2)}\xspace}

\newtheorem{thm}{Theorem}
\newtheorem{theorem}[thm]{Theorem}

\newtheorem{lemma}[thm]{Lemma}

\newtheorem{corollary}[thm]{Corollary}

\newtheorem{claim}[thm]{Claim}

\newtheorem{rem}[thm]{Remark}

\usepackage{authblk}

\title{Reconfiguration of graphs with connectivity constraints}
\date{\today}
\author[1]{Nicolas Bousquet\footnote{nicolas.bousquet@grenoble-inp.fr}}
\author[2]{Arnaud Mary\footnote{arnaud.mary@univ-lyon1.fr}}
\affil[1]{CNRS, G-SCOP, Grenoble-INP, Univ. Grenoble-Alpes, France. }
\affil[2]{LBBE, Universit\'e Claude Bernard Lyon 1, Lyon, France.}

\begin{document}

\maketitle

\begin{abstract}
A graph $G$ realizes the degree sequence $S$ if the degrees of its vertices is $S$.
Hakimi~\cite{HakiI} gave a necessary and sufficient condition to guarantee that there exists a connected multigraph realizing $S$. Taylor~\cite{taylor81} later proved that any connected multigraph can be transformed into any other via a sequence of flips (maintaining connectivity at any step). A flip consists in replacing two edges $ab$ and $cd$ by the diagonals $ac$ and $bd$.
In this paper, we study a generalization of this problem. A set of subsets of vertices $\mathcal{CC}$ is \emph{nested} if for every $C,C' \in \mathcal{CC}$ either $C \cap C' = \emptyset$ or one is included in the other. We are interested in multigraphs realizing a degree sequence $S$ and such that all the sets of a nested collection $\mathcal{CC}$ induce connected subgraphs. Such constraints naturally appear in tandem mass spectrometry.

We show that it is possible to decide in polynomial if there exists a graph realizing $S$ where all the sets in $\mathcal{CC}$ induce connected subgraphs. Moreover, we prove that all such graphs can be obtained via a sequence of flips such that all the intermediate graphs also realize $S$ and where all the sets of $\mathcal{CC}$ induce connected subgraphs. Our proof is algorithmic and provides a polynomial time approximation algorithm on the shortest sequence of flips between two graphs whose ratio depends on the depth of the nested partition.
\end{abstract}

\section{Introduction}
Let $G=(V,E)$ be a graph where $V$ denotes the set of vertices and $E$ the set of edges. All along the paper, unless otherwise specified, all the graphs are loop-free but may admit multiple edges. Reconfiguration problems consist in finding a step by step transformation between two  solutions of a given problem such that all intermediate states are also solutions. Reconfiguration problems arise in many different fields (e.g. graph theory~\cite{BonamyB13,CerecedaHJ11}, statistical physics~\cite{mohar4}, combinatorial games~\cite{HearnD05}, chemistry~\cite{Senior51} and peer-to-peer networks~\cite{Cooper07}) and received a considerable attention in the last few years.
For a complete overview of the reconfiguration field, the reader is referred to the recent surveys of van den Heuvel~\cite{Heuvel13} and Nishimura~\cite{Nishimura17}.
In this paper we consider the reconfiguration of graphs with a fixed degree sequence and its applications to cheminformatics.

The \emph{degree sequence} of a graph $G$ is the sequence of the degrees of its vertices in non-increasing order. Given a non-increasing sequence of integers $S=\{d_1,\ldots,d_n\}$, a graph $G=(V,E)$ with $V=\{v_1,\ldots,v_n\}$ \emph{realizes} $S$ if $d(v_i)=d_i$ for all $i\leq n$.

In the fifties, mathematicians tried to find conditions that guarantee that given a sequence of integers
$S=\{d_1,\ldots,d_n \}$, there exists a graph realizing $S$. Senior~\cite{Senior51} gave necessary and sufficient conditions for the case of connected (multi)graphs. Havel~\cite{Havel55} proposed a polynomial time algorithm that outputs a simple loop-free graph realizing $S$ if such a graph exists or returns no otherwise. Hakimi~\cite{HakiI} re-discovered the results of both Senior and Havel and  also proposed an algorithm that outputs a connected loop-free graph realizing $S$ if such a graph exists or returns no otherwise.



A \emph{flip} (also called \emph{swap} or \emph{switch} in the literature) on two edges $ab$ and $cd$ consists in deleting the edges $ab$ and $cd$ and creating the edges $ac$ and $bd$ (or $ad$ and $bc$)\footnote{In the case of multigraphs, we simply decrease by one the multiplicities of edges $ab$ and $cd$ and increase by one the ones of $ac$ and $bd$.}. The flip operation that transforms the edges $ab$ and $cd$ into the edges $ac$ and $bd$ is denoted $(ab,cd)\rightarrow (ac,bd)$. When the target edges are not important we will simply say that we flip the edges  $ab$ and $cd$.

Let $S=\{d_1,\ldots,d_n \}$ be a non-increasing sequence and let $G$ and $H$ be two graphs on $n$ vertices $v_1,\ldots,v_n$ realizing $S$. The graph $G$ can be \emph{transformed} into $H$ if there is a sequence of flips that transforms $G$ into $H$. Since flips do not modify the degree sequence, the intermediate graphs also realize $S$.
Let $\cG(S)$ be the graph whose vertices are loop-free multigraphs realizing $S$ and such that two vertices $G$ and $H$ of $\cG(S)$ are adjacent if $G$ can be transformed into $H$ via a single flip. Since the flip operation is reversible, the graph $\cG(S)$ is an undirected graph called the \emph{reconfiguration graph of $S$}. Note that there exists a sequence of flips between any pair of graphs realizing $S$ if and only if the graph $\cG(S)$ is connected.
In~\cite{HakiII}, Hakimi proved the following:

\begin{theorem}[Hakimi \cite{HakiII}]\label{thm:hakimiconn}
  Let $S$ be a non-increasing sequence. If the graph $\cG(S)$ is not empty, it is connected.
\end{theorem}

A connected reconfiguration graph has some interesting consequences for sampling or enumerating solutions. For instance, it implies that all the solutions can be enumerated with polynomial delay (as long as we get one of them). An enumeration algorithm is an algorithm that lists without repetition all the solutions of a given problem.
An algorithm solves an enumeration problem \emph{with polynomial delay}, if the delay between two consecutive outputs is bounded by a polynomial of the input size.
Any reconfiguration problem such that the reconfiguration graph is connected and the number of local operations (in our case flips) is polynomial admits a polynomial delay enumeration algorithm. So Theorem~\ref{thm:hakimiconn} ensures that there exists an algorithm that enumerates with polynomial delay all the graphs realizing $S$. Note however that the space needed by this algorithm might be exponential. As far as we know, the existence of a polynomial delay algorithm with polynomial space to generate all the graphs realizing $S$ is open.
\smallskip

One can wonder if the reconfiguration graph is still connected if only we consider graphs with additional properties.
For a graph property $\Pi$, let us denote by $\cG(S,\Pi)$ the subgraph of $\cG(S)$ induced by the
graphs realizing $S$ with the property
$\Pi$. If we respectively denote by $\mathscr{C}$ and $\mathscr{S}$ the property of being
connected and simple, Taylor proved in \cite{taylor81} that
$\cG(S,\mathscr{C})$, $\cG(S,\mathscr{S})$ and $\cG(S,\mathscr{C} \wedge
\mathscr{S})$ are connected ($\wedge$ stands for ``and'').
Let $G,H$ be two graphs of $\cG(S,\Pi)$. A sequence of flips \emph{transforms $G$ into $H$ in $\cG(S,\Pi)$} if the sequence of flips transforms $G$ into $H$ and all the intermediate graphs also have the property $\Pi$. Note that it is equivalent to find a path between $G$ and $H$ in $\cG(S,\Pi)$.


\paragraph{Applications to mass spectrometry.}

Mass spectrometry is a technique used to measure the mass-to-charge ($m/z$) ratio of molecules.
The process results in a $m/z$-spectrum, whose deconvolution provides a histogram with the quantity of each complex. Given this histogram, chemists can determine how many atoms of each type compose the molecule (i.e. the chemical formula of the molecule). With this chemical formula, we want to understand the structure of the molecule. Two question naturally arise: (i) Can we find a molecule structure satisfying this formula? (ii) Can we find all of them? Two molecules with the same chemical formula are called \emph{structural isomers}.

The problem of determining the structure of a molecule given its chemical formula, can be formulated as a combinatorial problem.
Let $v_1,\ldots,v_n$ be the $n$ atoms of the molecules. The degree of each atom $v_i$ is its valence. The questions then become: (i) Can we find a connected loop-free (multi-)graph on vertices $v_1,\ldots,v_n$ for which the degree of each $v_i$ is equal to the valence of its corresponding atom? (ii) If yes, can we generate all of them? As we have already seen, Hakimi~\cite{HakiII} and Taylor~\cite{taylor81} answered positively to both questions: we can enumerate with polynomial delay all the graphs in $\cG(S,\mathscr{C})$.

In the last few years, with the development of tandem mass spectrometry, we get more information on the structure of the original molecule. With this technology, we can again break the molecule into several fragments which in turn can be broken into other fragments...etc... For each produced fragment of this subdivision, we can determine its atoms constitution. Finally, we can obtain a tree of fragments, where each fragment corresponds to a part of the molecule that have to be connected. Rephrased in terms of graphs, it means that instead of simply knowing that the whole graph is connected, we are given a collection of subsets of vertices that have to induce connected subgraphs.
Since the number of graphs realizing a degree sequence is usually large, this additional information can drastically reduce the number of possible molecules.

\paragraph{Our results.}
 A collection $\mathcal{CC}$ of subsets of vertices is \emph{nested} if for every pair $C_i,C_j$ in $\mathcal{CC}$, either $C_i \cap C_j= \emptyset$ or one is included in the other. The \emph{height} $d$ of a nested partition $\mathcal{CC}$ is the maximum number of sets $C_{i_1},\ldots,C_{i_d}$ in $\mathcal{CC}$ such that $C_{i_1} \subsetneq C_{i_2} \subsetneq \ldots \subsetneq C_{i_d}$.

 Let $S=\{ d_1,\ldots,d_n\}$ be a degree sequence and $\mathcal{CC}$ be a nested collection such that $V \in \mathcal{CC}$. Let us denote by $\cG(S,\mathcal{CC})$ the subgraph of $\cG(S)$ induced by the graphs such that $G[C]$ is connected for $C$ in $\mathcal{CC}$. We study the three following questions:
 \begin{enumerate}[(i)]
  \item\label{q1} Is it possible to find in polynomial time a graph $G$ in $\cG(S,\mathcal{CC})$ if such a graph exists?
  \item\label{q2} Is $\cG(S,\mathcal{CC})$ a connected subgraph of $\cG(S)$?
  \item\label{q3} If yes, is it possible to find or approximate a shortest transformation between two graphs of $\cG(S,\mathcal{CC})$?
 \end{enumerate}

In Section~\ref{sec:realizability}, we answer positively to (\ref{q1}). We actually provide a necessary and sufficient condition for a graph to be realizable and then prove that this characterization can actually be turned into an algorithm.

In Section~\ref{sec:connectivity}, we answer to both (\ref{q2}) and (\ref{q3}). We show that, given two graphs $G,H$ in $\cG(S,\mathcal{CC})$, there always exists a transformation between $G$ and $H$ in $\cG(S,\mathcal{CC})$. To prove it, we exhibit an algorithm that finds a sequence of at most $(2 d+1) \delta(G,H)$ flips transforming $G$ into $H$, where $d$ is the height of the nested partition and $\delta(G,H)$ is the size of the \emph{symmetric difference} (see paragraph Notations for a formal definition).
Since the length of a minimum transformation between $G$ and $H$ is at least $\delta(G,H)/4$ (a flip decreases by at most four the size of the symmetric difference), we get an $(8d+4)$-approximation of the shortest sequence.
Note that it also provides as an immediate corollary a polynomial delay algorithm to enumerate all the graph in $\cG(S,\mathcal{CC})$.

\begin{theorem}\label{thm:main}
Let $\mathcal{CC}$ be a nested collection of subsets that contains $V$.
The graph $\cG(S,\mathcal{CC})$ is connected and the distance between any pair of graphs $G$ and $H$ in $\cG(S,\mathcal{CC})$ is at most $(2d+1)\delta(G,H)$.

Moreover, there exists a polynomial time algorithm that, given $G,H \in \cG(S,\mathcal{CC})$, computes a sequence of flips transforming $G$ into $H$ in $\cG(S,\mathcal{CC})$ of length at most $(8d+4) OPT$ where $OPT$ denotes the length of a shortest sequence.
\end{theorem}

We moreover show that finding a shortest sequence of flips between two graphs in $\cG(S,\mathscr{C})$ is NP-complete. It in particular implies as an immediate corollary that it is NP-complete for two graphs in $\cG(S,\mathcal{CC})$.
The proof follows the scheme of the NP-hardness proof of Will~\cite{Will99} in the case of simple graphs and is given in appendix.

In order to prove Theorem~\ref{thm:main}, we need as a black-box an approximation algorithm of the shortest transformation between two graphs in $\cG(S,\mathscr{C})$. In Section~\ref{sec:approx-classique}, we give an algorithm that provides a transformation from $G$ into $H$ in $\cG(S,\mathscr{C})$ of size at most four times the optimal one for any pair of graphs $G,H$ in $\cG(S,\mathscr{C})$. This result also provides an alternative proof of the result of Taylor~\cite{taylor81}.

\medskip

\noindent\textbf{Notations.}\\
All along the paper, we consider unoriented loop-free multigraphs. Given two graphs $G$ and $H$ on the same vertex set $V$, we denote by $G\Delta H$ their symmetric  difference i.e. the (multi)set of edges such that $e$ appears in $G\Delta H$ with multiplicity $r>0$ if the difference between the multiplicities of $e$ in $G$ and in $H$ is equal to $r$ or $-r$. We denote by $\delta(G,H)$ the size of $G\Delta H$. By abuse of notations, we often assimilate $G \Delta H$ to the graph $G=(V,G\Delta H)$. We denote by $G -H$ the (multi)set of edges such that $e$ appears in $G - H$ with multiplicity $r>0$ if the difference between the multiplicities of $e$ in $G$ and in $H$ is equal to $r$. For simple graphs it simply corresponds to the edges which are in $G$ and not in $H$. We also assimilate $G -H$ to the graph $(V,G-H)$. Note that $\delta(G,H)$ is twice the number of edges in $G - H$.
Finally let $G \cap H$ be the set of edges containing $e$ is an edge with multiplicity $r$ if the minimum multiplicity of $e$ in $G$ and $H$ is exactly $r$. As for $G - H$ and $G \Delta H$, we assimilate $G \cap H$ to $(V,G \cap H)$. Note that $E(G) = E(G \cap H) \cup E(G-H)$.

\section{$4$-Approximation for connected graphs}\label{sec:approx-classique}
Let $S$ be a non-increasing degree sequence. In~\cite{taylor81}, Taylor proved that $\cG(S,\mathscr{C})$ is connected. However, his proof does not immediately provide an approximation algorithm of the shortest transformation between pairs of graphs in $\cG(S,\mathscr{C})$. In this section we give an alternative proof of the result of Taylor that provides a $4$-approximation algorithm of the shortest transformation between any pair of graphs in $\cG(S,\mathscr{C})$.

\begin{theorem}\label{thm:connected}
Let $S$ be a non-increasing degree sequence and $G$ and $H$ in $\cG(S,\mathscr{C})$. We can find in polynomial time a sequence of at most $\delta(G,H)$ flips transforming $G$ into $H$ in $\cG(S,\mathscr{C})$. \\
Moreover this transformation never flips any edge that is already in both $G$ and $H$ \footnote{We say that an edge $e$ in $G \cap H$ is never flipped if the multiplicity of $e$ at any intermediate step never goes below the multiplicity of $e$ in $G \cap H$.}.
\end{theorem}

The size of a shortest transformation is at least $\delta(G,H)/4$ since at most two edges of $G$ can be flipped on edges of $H$ at every step. So Theorem~\ref{thm:connected} provides a $4$-approximation algorithm. The remaining of this section is devoted to prove the following lemma whose iterated application immediately implies Theorem~\ref{thm:connected}. We say that a flip \emph{maintains connectivity} (of a connected graph $G$) if the resulting graph after the flip is still connected. A sequence of flips maintains connectivity if all intermediate graphs are connected.

\begin{lemma}\label{lem:4approxmain}
Let $S$ be a non-increasing degree sequence and $G$ and $H$ in $\cG(S,\mathscr{C})$. There exists a sequence of at most two flips in maintaining connectivity that decreases by at least two the size of the symmetric difference. Moreover the sequence never flips any edge that is already in both $G$ and $H$.
\end{lemma}

\begin{proof}
In order to prove it, let us first prove that there exist cases where we can easily decrease the symmetric difference in one step.

\begin{claim}\label{f1}
If an edge of $G-H$ is contained in a cycle of $G$ \footnote{Two parallel edges form a cycle of length $2$.}, then there exists a flip maintaining the connectivity that decreases by at least two the size of the symmetric difference.
\end{claim}
\begin{proof}
Let $e=uv$ be an edge of $G-H$ which is contained in a cycle of $G$.

Let us first prove that there exists an edge $wx$ of $G-H$ with $w,x$
different from $u,v$ such that either $uw$ or $vx$ is an edge of $H - G$.
Assume for contradiction that for all edges $uw$ (resp. $vx$) in $H-G$ the only edges incident to $w$ (resp. $x$) in $G-H$ are incident to $v$ (resp. $u$). Let $N$ be the neighborhood of $v$ $H-G$. Since the degree of any vertex in $G-H$ is equal to the one in $H-G$  the sum of the degrees of the vertices of $N$ in $G-H$ is at least $deg_{H-G}(v)$. Since we assumed that any edge incident to $N$ in $G-H$ is also incident to $u$, and since $uv$ is an edge of $G-H$, we have $deg_{G-H}(u)\geq deg_{H-G}(v) + 1$. Thus $deg_{H-G}(u)\geq deg_{H-G}(v)+1$. Using the same arguments symmetrically, we have  $deg_{H-G}(v)\geq deg_{H-G}(u)+1$.

So there exists an edge $wx$ of $G-H$ with $w,x$ different from $u,v$ such that either $uw$ or $vx$ is an edge of $H-G$.

Then the flip $(uv,wx) \rightarrow (uw, vx)$  decreases by $2$ the size of the size of the symmetric difference and the resulting graph is still connected since $e$ was in a cycle of $G$. Moreover, since $w$ and $x$ are distinct from $u,v$, the resulting graph is loopless.
\end{proof}

If Claim~\ref{f1} can be applied, Lemma~\ref{lem:4approxmain} holds. So we can assume that all the edges of $G-H$ do not belong to a cycle of $G$, i.e. all of them are bridges.
Let $G'$ be the graph obtained from $G$ by contracting all the connected components of $G\cap H$ into a single vertex. And there is an edge $S_1S_2$ in $G'$ if there is an edge $u_1u_2$ of $G$ with $u_1$ in $S_1$ and $u_2$ in $S_2$.  We say that $v_1v_2$ \emph{corresponds to} $S_1S_2$ in $G'$. Note that by Claim~\ref{f1} there is a bijection between the edges of $G-H$ and the edges of $G'$.
Since no edge of $G-H$ is contained in a cycle of $G$, the graph $G'$ is a tree (without multiedges).

\begin{figure}
 \centering
 \includegraphics[scale=0.5]{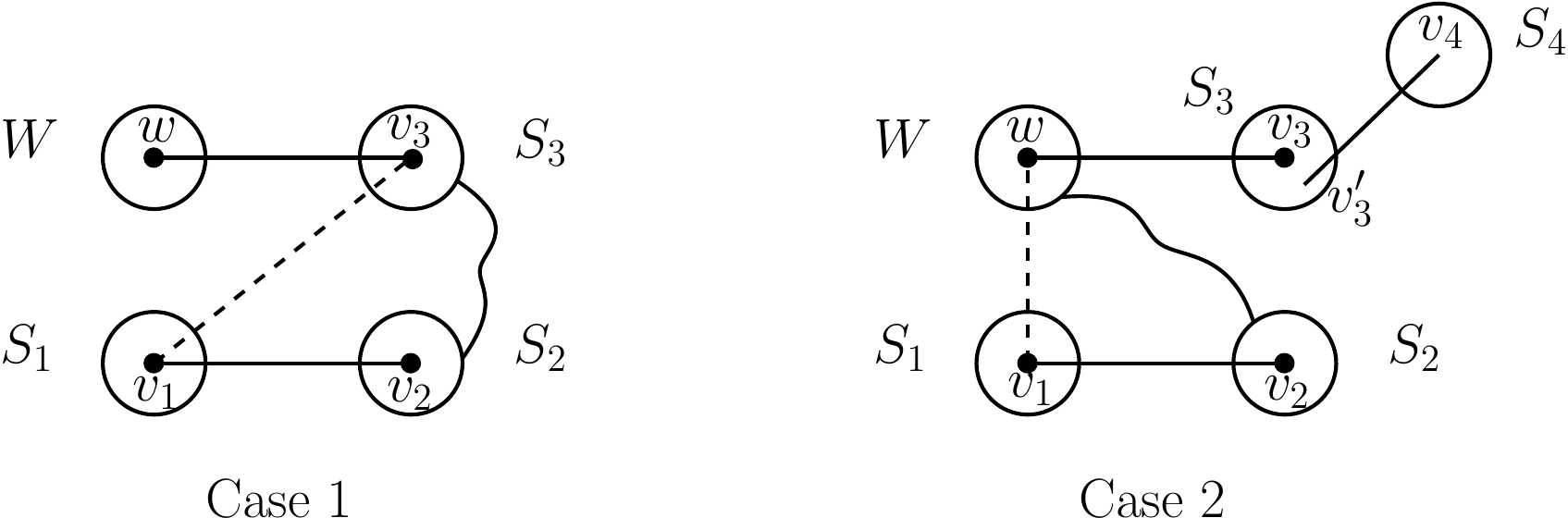}
 \caption{Illustration of the proof of Lemma~\ref{lem:4approxmain}. Full-edges are edges of $G -H$, dotted edges are edges of $H - G$ and zigzags represent paths in $G$}
 \label{fig:4approx}
\end{figure}

Let $S_1S_2$ be an edge of $G'$ such that $S_1$ is a leaf of $G'$ and let $v_1v_2$ be the corresponding edge in $G-H$ with $v_1\in S_1$ and $v_2\in S_2$. Now let $v_1v_3$ be the edge of $H-G$ incident to $v_1$ and let $S_3$ be the vertex of $G'$ containing $v_3$ (since an edge exists since $G$ and $H$ have the same degree sequence). 
Since $G$ and $H$ have the same degree sequence, there exists an edge $v_3w$ in $G-H$. Let $W$ be the set of $G\cap H$ containing $w$. An illustration of the proof is proposed in Figure~\ref{fig:4approx}.\smallskip

\noindent\textbf{Case 1.} $W$ is not in the path between $S_1$ and $S_3$ in $G'$. \\
The flip $(v_1v_2,v_3w) \rightarrow (v_1v_3,v_2w)$  does not disconnect $G$. Moreover, the symmetric difference decreases by~$2$.
\smallskip

\noindent\textbf{Case 2.} $W$ is in the path between $S_1$ and $S_3$ in $G'$. \\
In this case we cannot simply perform the flip  $(v_1v_2,v_3w) \rightarrow (v_1v_3,v_2w)$ since it would disconnect the graph $G$. First note that $S_1 \cup S_3 \ne V(G')$ (since otherwise both $S_1$ and $S_3$ would have degree one).
We claim that $S_3$ is not a leaf of $G'$. Assume by contradiction that $S_1$ and $S_3$ are leaves of $G'$ and $v_1v_3 \in E(G)$. Then the only one edge of $G$ with one endpoint in $S_1$ and one endpoint in $V \setminus S_1$ has an endpoint in $S_3$. Similarly, the one edge of $G$ with one endpoint in $S_3$ and one endpoint in $V \setminus S_3$ has an endpoint in $S_1$. Since $S_1 \cup S_3 \ne V$, the graph $G$ is not connected, a contradiction.

So $S_3$ has degree at least $2$ in $G'$. Let $v'_3v_4$ be an edge of $G-H$ distinct from $v_3w$ with one endpoint $v_3'$ in $S_3$ and one endpoint $v_4$ in $V \setminus S_3$ (note that $v_3'$ might be $v_3$). Let $S_3S_4$ be be the corresponding edge in $G'$ with $v'_3\in S_3$ and $v_4\in S_4$.
We first perform the flip $(v_1v_2,v'_3v_4) \rightarrow (v_1v'_3,v_2v_4)$ in $G$. Notice that this flip does not disconnect the graph since $S_4$ is not on the pat between $S_1$ and $S_2$. Since both edges $v_1v_2$ and $v'_3v_4$ are in $G-H$, the symmetric difference does not increase and no edge of $G \cap H$ is flipped. Note moreover that if $v_3 =v_3'$ then the symmetric difference has decreased and we are done. So we can assume that $v_3 \ne v_3'$.
We then perform the flip $(v_1v'_3,v_3w) \rightarrow (v_1v_3,v'_3w)$. This flip does not disconnect the graph (since $S_3$ induces a connected subgraph) and decreases by two the symmetric difference (since it creates $v_1v_3$), which completes the proof.
\end{proof}

\section{Tree of the fragments}\label{sec:rooted_tree}

\begin{figure}[h]
\fbox{
\begin{subfigure}{0.3\textwidth}
        \centering
        \includegraphics[scale=0.6]{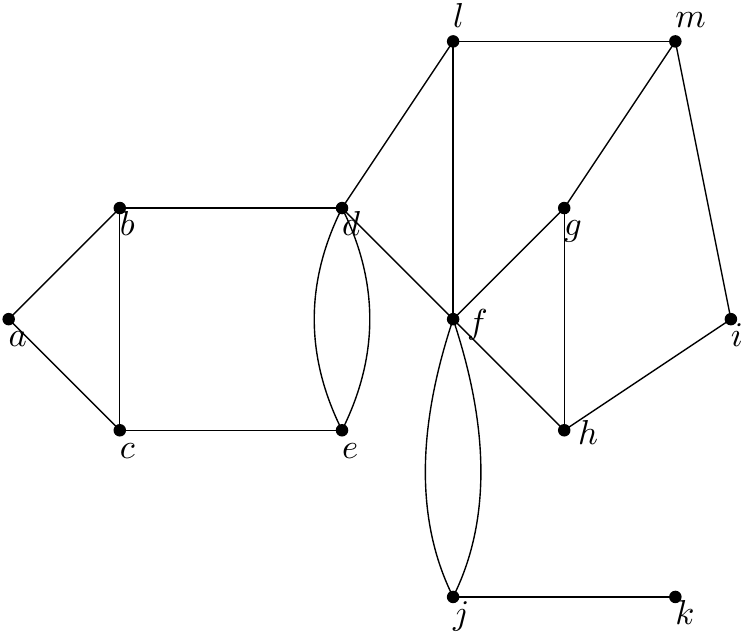}
        \caption{A graph $G$}
        \label{fig:sfig2}
    \end{subfigure}
}
\fbox{
\begin{subfigure}{.33\textwidth}
        \centering
        \includegraphics[scale=0.6]{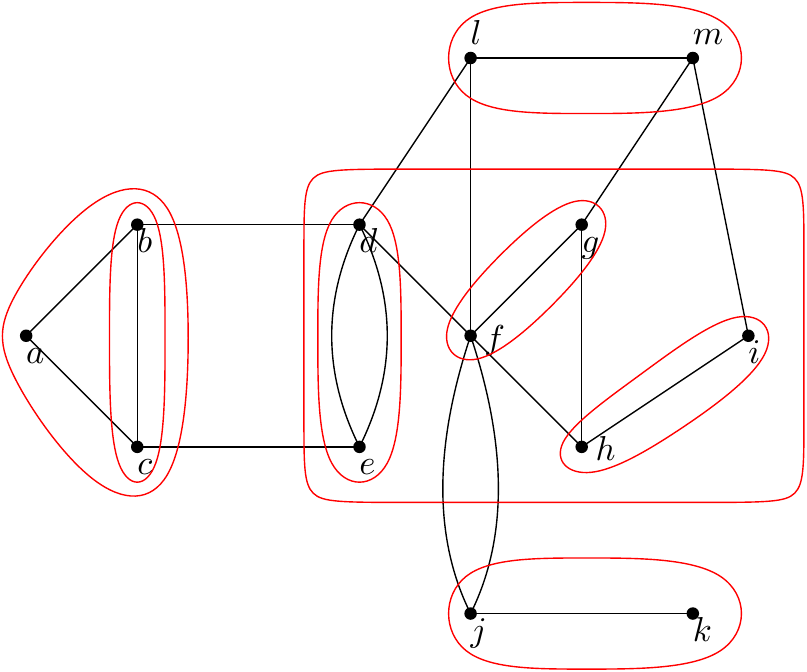}
        \caption{$G$ and the sets of the collection $\mathcal{CC}$}
        \label{fig:sfig4}
    \end{subfigure}}
     \fbox{
   \begin{subfigure}{.32\textwidth}
        \centering
        \includegraphics[scale=0.6]{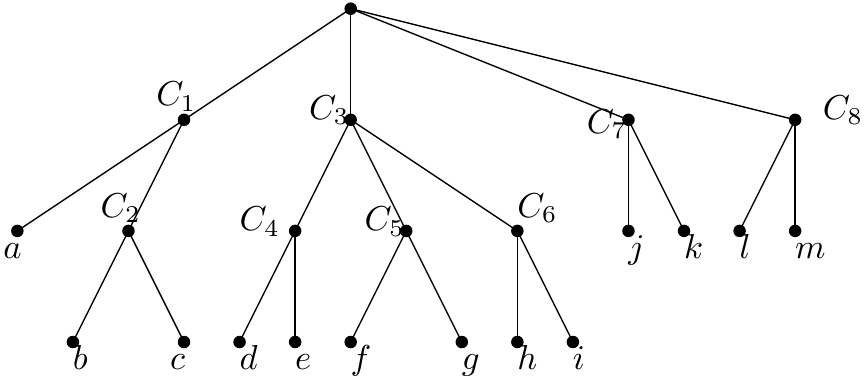}
        \caption{The tree of the fragments}
        \label{fig:sfig3}
    \end{subfigure}}
\caption{Example of a graph $G$ with a nested partition}
\label{fig:exemple1}
\end{figure}

Let $V$ be a vertex set and $S$ be a degree sequence of size $|V|$. Let $\mathcal{CC}$ be a nested collection of subsets of $V$ such that $V$ and all the singletons belong to $\mathcal{CC}$. Singletons are included in $\mathcal{CC}$ for convenience since their addition does not change the graph $\cG(S,\mathcal{CC})$. Indeed, a single vertex induces a connected subgraph.

Let $G$ and $H$ be two graphs in $\cG(S,\mathcal{CC})$. Let $e=uv$ be an edge with multiplicity $r_1$ in $G$ and $r_2$ in $H$ and let $r=\min(r_1,r_2)$. Then $r$ copies of $e$ are \emph{good} and the others are bad. (In case of simple graph an edge $e$ of $G$ is good if it is also an edge of $H$ and it is bad otherwise).
A flip in $G$ is \emph{correct} if it maintains the connectivity of $G[C]$ for any $C \in \mathcal{CC}$. A flip \emph{does not modify a good edge $e$} if the multiplicity of $e$ after the flip is still at least the multiplicity of $e$ in $G \cap H$.

Let $\mathcal{CC}$ be a nested collection. The \emph{tree of the fragments} $T$ is the tree whose nodes are labeled by elements of $\mathcal{CC}$ and there is an arc from $C_1$ to $C_2$ if $C_2 \subseteq C_1$ and there does not exist $C \in \mathcal{CC}$ distinct from $C_1$ and $C_2$ such that $C_2 \subseteq C \subseteq C_1$. In other words, $T$ is the tree rooted at $V$ corresponding to the nested partition of $\mathcal{CC}$ (see Figure~\ref{fig:exemple1} for an illustration).
By abuse of notation and when no confusion is possible, $C$ will denote both the node of $T$ and the corresponding set in $\mathcal{CC}$. Since $\mathcal{CC}$ contains all the sets of size one, the leaves of $T$ are the vertices of $V$. Since $\mathcal{CC}$ is nested, $T$ is well-defined and is a tree.
Given a node $C$ of the tree of the fragments and $v \in V$, $v \in C$ if and only if the leaf labeled with $v$ is a leaf of the subtree rooted at $C$.
We denote by $G[C]$ the subgraph of $G$ induced by the vertices in $C$.

\begin{figure}[h]
        \centering
        \includegraphics[scale=1]{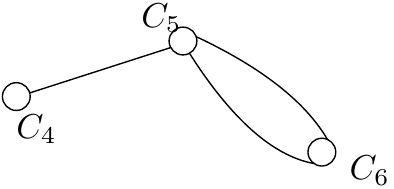}
\caption{The graph \contracted{G}{C_3}}
        \label{fig:sfig5}
\end{figure}

Given a node $C$, we denote by $ch(C)$ the children of $C$ in $T$.
Let $C$ be an internal node of $T$ and $G \in \cG(S,\mathcal{CC})$. We denote by $\contracted{G}{C}$, the graph with vertex set $ch(C)$ where $C'$ and $C''$ in $V(\contracted{G}{C})$ are adjacent with multiplicity $k$ if there exist exactly $k$ edges with one endpoint in $C'$ and one endpoint in $C''$ in $G$.
See Figure~\ref{fig:sfig5} for an illustration.

Note that there is a natural bijection between edges of $G$ and edges of $\contracted{G}{C}$.
Indeed, for any edge $e$ of $G$, there exists a unique $C \in \mathcal{CC}$ in which an edge is created in $\contracted{G}{C}$ because of $e$. The following remark follows from that observation.

\begin{rem}\label{rem:SafeOutside}
Let $G$ be a graph in $\cG(S,\mathcal{CC})$.
\begin{enumerate}
 \item For every $C \in T$, $|E(G[C])| = \bigcup_{C' \subseteq C} |E(\contracted{G}{C'})|$;
 \item Let $C\in T$ and let $e\in E(\contracted{G}{C})$, then $\contracted{G}{C'}-e$ is connected for every $C'\in T$, $C'\neq C$.
\end{enumerate}
\end{rem}

The second point holds since $e$ has no impact on $\contracted{G}{C'}$ for $C' \neq C$.
Let us first prove the following lemma that will be used all along the proof.

\begin{lemma}\label{lem:ContractEquiv}
  Let $G$ be a  graph in $\cG(S)$.
  \begin{itemize}
   \item If $G[C]$ is connected, then $\contracted{G}{C}$ is connected.
   \item Let $C$ be a node of $T$ and $T'$ be the subtree rooted at $C$. If $\contracted{G}{C'}$ is connected for every node $C'$ of $T'$ then $G[C']$ is connected for every node $C'$ of $T'$.
  \end{itemize}
\end{lemma}

\begin{proof}
If $G[C]$ is connected, then $\contracted{G}{C}$ is connected since $\contracted{G}{C}$ is obtained from $G[C]$ by identifying vertices. So the first point holds.

Assume now that $\contracted{G}{C'}$ is connected for every $C'\in T'$. Let us prove by induction bottom-up from the leaves that the right part of the second point holds. A leaf indeed induces a connected subgraph. Now let $C'$ be an internal node. By induction, for any child $C''$ of $C'$ in $T'$, the graph $G[C'']$ is connected. Moreover by hypothesis the graph  $\contracted{G}{C'}$ is connected. Since vertices of $\contracted{G}{C'}$ correspond to the identification of $G[C'']$ (which is connected) for every child $C''$ of $C'$, the graph  $\contracted{G}{C'}$ is connected.
\end{proof}

The second point will be applied in several situations since manipulating $\contracted{G}{C}$ is often simpler than manipulating $G[C]$.
Note that when $C$ is the root, the conclusion of the second point ensures that $G$ is in $\cG(S,\mathcal{CC})$.


Let $V_1,V_2$ be a partition of $V$. Then $E_G(V_1,V_2)$ denotes the set of edges of $G$ with one endpoint in $V_1$ and one endpoint in $V_2$.

\begin{lemma}\label{lem:GoodBadEdge}
     Let $G,H$ be two graphs of $\cG(S,\mathcal{CC})$ and let $C \in \mathcal{CC}$.
     If $|E_{G}(C,V \setminus C)| < |E_{H}(C,V \setminus C)|$, then there exists $e \in G[C] - H[C]$ such that $G[C'] -e$ is connected for every node $C'$ of $T$.
\end{lemma}

Note that when we say ``\textit{there exists $e \in G[C] - H[C]$}'' in the statement of Lemma~\ref{lem:GoodBadEdge}, the edge $e$ might exist in both $G$ and $H$, but in that case the multiplicity of $e$ in $H$ has to be strictly smaller than the multiplicity of $e$ in $G$.

\begin{proof}
Since $G$ and $H$ have the same degree sequence, we have $\sum\limits_{x \in C} deg_{G}(x)=\sum\limits_{x \in C} deg_{H}(x)$. Since $|E_{G}(C,V \setminus C)| < |E_{H}(C,V \setminus C)|$, we have $|E(G[C])|> |E(H[C])|$.
By Remark~\ref{rem:SafeOutside},
\[
    |E(G[C])|= \sum_{C' \subseteq C, \ C' \in \mathcal{CC}} |E(\contracted{G}{C'})|> \sum_{C' \subseteq C, \ C' \in \mathcal{CC}} |E(\contracted{H}{C'})| = |E(H[C])|
\]

So there exists $C' \subseteq C$ with $C' \in \mathcal{CC}$ such that $|E(\contracted{G}{C'})|> |E(\contracted{H}{C'})|$.

A \emph{block} is a a connected component of $\contracted{G}{C'} \cap \contracted{H}{C'}$. Let $b$ be the number of blocks. Since $\contracted{H}{C'}$  is connected (since $C' \in \mathcal{CC}$ and by Lemma~\ref{lem:ContractEquiv}), there are at least $b-1$ edges of $\contracted{H}{C'}$ which are not in $\contracted{G}{C'} \cap \contracted{H}{C'}$. Since by assumption $\contracted{G}{C'}$ contains more edges than $\contracted{H}{C'}$, the number of edges in $\contracted{G}{C'} - \contracted{H}{C'}$ is at least $b$. So then there is an edge $e_1$ of $\contracted{G}{C'} - \contracted{H}{C'}$ with both endpoints in the same block or there is a cycle of edges of $\contracted{G}{C'} - \contracted{H}{C'}$ between the blocks. The deletion of $e_1$ in the first case or any edge of the cycle in the second case leaves the graph $\contracted{G}{C'}$ connected. Moreover, for every $C''$ in $T$ distinct from $C'$, the graph $\contracted{G}{C''}$ is not modified by Remark~\ref{rem:SafeOutside}, and then is still connected.
So the second point of Lemma~\ref{lem:ContractEquiv} ensures that the conclusion holds.
\end{proof}

\section{Realizability}\label{sec:realizability}

Let $V=\{v_1,\ldots,v_n\}$ be a set of vertices and let $\mathcal{CC}$ be a nested partition containing $V$ and all the singletons. Let $T$ be the tree of the fragments of $\mathcal{CC}$. In this section, we provide a necessary and sufficient condition on the degree sequence $S=\{d_1,\ldots,d_n \}$ to be realizable by a loop-free multigraph such that $C$ induces a connected subgraph for every $C \in \mathcal{CC}$. This characterization generalizes the ones of~\cite{HakiI} and \cite{Senior51} for loop-free multigraphs since in this case $\mathcal{CC} = \{ V \}$. We end the section by explaining how this characterization can be turned into a polynomial time algorithm.

Let $C$ be a node of $T$. A graph $G'$ is \emph{coherent on $C$} if it is defined on $C$ and:
\begin{itemize}
 \item for every $v_i \in C$, $d_{G'}(v_i) \le d_i$ and,
 \item for every $C' \in \mathcal{CC}$ such that $C' \subseteq C$, $G'[C']$ is connected.
\end{itemize}
Let $G'$ be a graph \emph{coherent} on $C$. The \emph{degree-deficit} of $C$ (for $G'$) is equal to $\sum_{v_i \in C} (d_i - d_{G'}(v_i))$. In other words, the degree-deficit represents the amount of endpoints of edges ``missing'' to complete the degree of the vertices of $C$. Note that if a graph $G$ defined on $V$ realizes $T$, then the degree-deficit of $C$ is the number of edges with one endpoint in $C$ and one endpoint in $V \setminus C$. Moreover for any node $C$, the graph $G[C]$ is coherent on $C$. Let us start with a straightforward remark.

\begin{rem}\label{rem:u}
Let $C'$ and $C''$ be two disjoint sets in  $\mathcal{CC}$. Let $G'$ be a graph defined on $C' \cup C''$. If the degree-deficits of both $C'$ and $C''$ are positive, then there exists $u \in C'$ and $v \in C''$ such that the edge $uv$ can be added in $G'$ without violating any degree constraint.
\end{rem}

We now define $\ell(C)$ and $u(C)$. We will then prove that, given a graph $G$ in $\cG(S,\mathcal{CC})$, they correspond to respectively the minimum and the maximum degree-deficit of $C$ for any graph $G'$ coherent on $C$.
For any node $C$ in $\mathcal{CC}$, we define

\[u(C):= \sum \limits_{v_i\in C }d_i - (2|C|-2)\]

\begin{equation*}
\text{and} \ \ \ \ \ \
\ell(C) = \begin{cases}
d_i & \text{if } C \text{ is the leaf } v_i\\
\varphi (C) & \text{if } \varphi(C)\geq 0\\
0 & \text{if } \varphi(C)<0\text{ and }\varphi(C) \text{ is even}\\
1 & \text{if } \varphi(C)<0\text{ and }\varphi(C) \text{ is odd}
\end{cases}
\end{equation*}
\[\text{where } \ \ \ \ \ \varphi(C)= \max\limits_{C_j \in ch(C)} \Big( \ell(C_j) - \sum_{\substack{C_i \in ch(C)\\
C_i\neq argmax\{\ell(C_j)\}}}u(C_i) \Big).\]

\begin{lemma}\label{lem:necessary_cond}
Let $S$ be a degree sequence and $\mathcal{CC}$ be a nested partition. Let $G$ in $\cG(S,\mathcal{CC})$. Then for every node $C$ of the tree of the fragments, the degree-deficit $s(C)$ of $C$ satisfies:
 \[ \ell(C) \le s(C) \le u(C) \]
 Moreover, $\ell(C)$ and $u(C)$ are even if and only if $\sum_{v_i \in C} d_i$ is even.
\end{lemma}
\begin{proof}
The number of edges of a connected subgraph on $n$ vertices is at least $n-1$. Since $G[C]$ is connected, there are at least $|C|-1$ edges with both endpoints in $C$ and then we have $s(C) \le u(C) = \sum \limits_{v_i\in C }d_i - (2|C|-2)$. Since we removed an even value from $\sum \limits_{v_i\in C }d_i$, $u(C)$ is even if and only if $\sum \limits_{v_i\in C }d_i $ is even.

Let us now prove that $s(C) \ge \ell(C)$. We prove it by induction bottom-up from the leaves. If $C$ is a leaf, then $C=\{v_i\}$ and there are exactly $d_i$ edges between $v_i$ and its complement in $G$. Thus $s(C)=E(C,V\setminus C) = \ell(C)$ and the parity of $s(C)$ is indeed the one of $\ell(C)$.

Now, let $C$ be an internal node. By induction hypothesis, for every child $C_i$ of $C$, the parities of $\ell(C_i)$ and $u(C_i)$ are the parity of $\sum \limits_{v_i\in C_i }d_i $. Thus, by definition of $\ell(C)$, the parity of $\ell(C)$ is the parity of $\sum \limits_{v_i\in C }d_i $.
So, in particular, if  $\varphi(C) \le 1$, the conclusion holds since the degree-deficit cannot be negative.
So we can assume that $\varphi(C) \ge 2$. Let us denote by $C_1,\ldots,C_r$ the children of $C$ and we can assume w.l.o.g. that $C_1$ is the child of $C$ satisfying $\varphi(C)= \ell(C_1) - \sum_{i \ge 2} u(C_i)$. Let $N:=\sum_{i \ge 2} u(C_i)$. The first part of the proof ensures that the degree-deficit of $\cup_{i \ge 2} C_i$ is at most $N$. Moreover, by induction, the degree-deficit of $C_1$ is at least $\ell(C_1)$. So the maximum number of edges between $C_1$ and $\cup_{i \ge 2} C_i$ in a graph satisfying all the constraints is $N$ (since $\ell(C_1) > N$). So the degree-deficit of $C$ is at least $\ell(C_1)-N$ which completes the proof.
\end{proof}





The goal of this section consists in proving the following theorem that ensures that it suffices to look at the values $u(C)$ and $\ell(V)$ to determine if a graph is realizable.

\begin{theorem}\label{thm:rea}
Let $S$ be a degree sequence and $\mathcal{CC}$ be a nested partition containing $V$. There exists a graph in $\cG(S,\mathcal{CC})$ if and only if:
 \begin{enumerate}
  \item For every internal node $C$ of the tree of the fragments, $u(C) \ge 1$.
  \item $\ell(V)=0$ and $u(V) \ge 0$.
 \end{enumerate}
\end{theorem}

Let $G$ be a graph in $\cG(S,\mathcal{CC})$. Then the degree-deficit of $V$ equals $0$. So by Lemma~\ref{lem:necessary_cond} applied on $V$, $\ell(V)=0$ and $u(V) \ge 0$ is necessary. Moreover, since $G$ is connected, at least one edge has to have one endpoint in $C$ and one endpoint in $V \setminus C$ for every $C \subsetneq V$. Thus the degree-deficit of $C$ is at least one for every $C \in \mathcal{CC}$, $C \neq V$, and then Lemma~\ref{lem:necessary_cond} ensures that the first condition is necessary. To prove Theorem~\ref{thm:rea}, we have to show that these two conditions are sufficient. The sufficiency is an immediate corollary of the next lemma applied to $V$ with $s=\ell(V)=0$.

\begin{lemma}\label{lem:decrease}
Let $S$ be a degree sequence and $\mathcal{CC}$ be a nested partition. Let $C$ be a node of the tree of the fragments $T$ such that $u(C) \ge 0$ and, for every $C' \subsetneq C$, $u(C') \ge 1$. For every $s$ such that $\ell(C) \le s \le u(C)$ and such that $s$ has the same parity as $\ell(C)$ and $u(C)$, there exists a graph $G'$ coherent on $C$ with degree-deficit~$s$.
\end{lemma}

In order to prove Lemma~\ref{lem:decrease}, we need to show the following lemma as an intermediate step.

\begin{lemma}\label{lem:uc}
Let $S$ be a degree sequence and $\mathcal{CC}$ be a nested partition. Let $C$ be a node of the tree of the fragments $T$ such that $u(C) \ge 0$ and such that, for every $C' \subsetneq C$ we have $u(C') \ge 1$. Then there exists graph $G'$ coherent on $C$ with degree-deficit $u(C)$.
\end{lemma}

\begin{proof}
We prove it by induction bottom-up from the leaves. For leaves, we have $u(C)=u(v_i)=d_i$ and then the conclusion holds. Let us now assume that $C$ is an internal node and $C_1,\ldots,C_r$ be the children of $C$ in $T$. By induction there exist subgraphs $G_1,\ldots,G_r$ coherent on respectively $C_1,\ldots,C_r$ such that the degree-deficit $s_i$ of $G_i$ is $u(C_i)$ for every $i \le r$.  Since $u(C_i) \ge 1$ for every $i$ and $u(C) \ge 0$, we can connect all these subgraphs using $r-1$ edges. To prove it, let us create an auxiliary graph $H$ with initial vertex set $1,\ldots,r$ and with degree constraint $d_i:=u(C_i)$ for vertex $C_i$. Note that all the degree constraint are positive since $u(C_i)>0$. Now, since $u(C) \ge 0$, the sum of the degree constraints in $H$ is at least $2(r-1)$. In particular, there exists a vertex $i$ of $H$ of degree $d_i \ge 2$. Let $j$ be a vertex of minimum degree. Create an edge between $C_i$ and $C_j$ in $G'$ (it is possible by definition of $d_i$ and $d_j$); And in $H$ contract the two vertices $i$ and $j$ into a single one of degree constraint $d_i+d_j -2$. Let us still denote by $H$ the resulting graph. Note that the graph induced by $C_i \cup C_j$ now induces a connected graph. Moreover, the graph $H$ still satisfies $\sum d_i \ge 2(|V(H)|-1)$ and all its degree constraints are positive. So we can repeat the operation until it only remains one vertex. At that point the graph $G'$ is connected on $C$. Furthermore its degree-deficit is $\sum\limits_{1}^{r} u(C_i) - 2(r-1) = \sum\limits_{v_i\in C} d_i - 2(|C| - 1 )=u(C)$ by Remark~\ref{rem:u} which concludes the proof.
\end{proof}

Using Lemma~\ref{lem:uc}, we can now prove Lemma~\ref{lem:decrease}.

\begin{proof}[Proof of Lemma~\ref{lem:decrease}]
 We prove it bottom-up from the leaves of the tree of the fragments. If $C$ is a leaf, then $C=\{ v_i \}$ and the degree deficit $s$ of $C$ is always $d_i=u(C)=\ell(C)$.

 Let $C$ be an internal node.   When $s=u(C)$, the conclusion holds by Lemma~\ref{lem:uc}. Let us prove that if there is a graph $G$ coherent on $C$ with degree-deficit $s$, then there exists a graph $G'$ coherent on $C$ with degree-deficit $s-2$ as long as $s-2 \ge \ell(C)$.
 Let us denote by $C_1,\ldots,C_r$ the children of $C$ in the tree of the fragments.

 If there are two vertices $v_i,v_j$ of $G$ such that $d_{G}(v_i) <d_i$ and $d_G(v_j) < d_j$, then we can create an edge $v_iv_j$ and the resulting graph indeed has degree-deficit $s-2$ since $d_G(v_i)-d_i$ and $d_G(v_j)-d_j$ decrease by one (and the other differences remain unchanged). So we can assume that there exists a vertex $v_i$ such that $d_i-d_G(v_i)=s$.  Let us denote by $C_i$ the child of $C$ in the tree of the fragment containing $v_i$.

If the graph $\contracted{G}{C}$ is not a star, then there is an edge $uv$ with one endpoint in $C_j \in ch(C)$ and one endpoint in $C_{j'} \in ch(C)$ with $j,j' \ne i$. Then, delete $uv$ and create $v_iu$ and $v_iv$. Note that all the graphs $\contracted{G}{C'}$ for $C' \subseteq C$ remain connected. Moreover $d_i-d(v_i)$ decreases by two and the other differences are not modified.
So from now on we assume that $\contracted{G}{C}$ is a star. Using similar arguments, we can moreover prove that all the edges with one endpoint in $C_i$ and one endpoint in $C' \ne C_i$ have endpoint $v_i$.
We now distinguish two cases.  \smallskip

\noindent\textbf{Case 1.} There exists $C_j \ne C_i$  of degree-deficit $s_j < u(C_j)$. \\
By induction, there exists a graph $G'_j$ defined on $C_j$ with degree-deficit is $s_j+2$. Let us now define a graph $G'$ on $C$ with degree-deficit $s-2$. The graph $G[C_a]=G'[C_a]$ for $a \ne j$ and $G[C_j]=G'_j$. Now let $v_iu$ be an edge of $G$. If $u \notin C_j$, we create the same edge in $G'$. Now let $u_1,\ldots,u_{s_j+2}$ be the vertices of $C_j$ missing at least one edge (with multiplicity). We create the set of edge $v_iu_1,\ldots,v_iu_{s_j+2}$. The graph $G'$ satisfies all the constraints and the degree-deficit of $G'$ is the degree-deficit of $G$ minus $2$.\smallskip


\noindent\textbf{Case 2.} For every $j \ne i$ the degree-deficit is equal to $u(C_j)$. \\
Since $ \ell(C_i) - \sum_{j \neq i}u(C_j) \le s-2$ and since $\contracted{G}{C}$ is a star rooted at $C_i$, the degree-deficit $s_i$ of $C_i$ is strictly larger than $\ell(C_i)+2$. By induction, there exists a graph $G'_i$ defined on $C_i$ with degree-deficit is $s_i-2$. Let us now define a graph $G$ on $C$ with degree-deficit $s-2$. The graph $G[C_a]=G'[C_a]$ for $a \ne i$ and $G[C_i]=G'_i$. Now let $w_1,\ldots,w_r$ be the neighbors of $v_i$ in $C \setminus V_i$ in the graph $G$. Since the degree-deficit of $G[C_i]$ was at least $r+2$, the graph $G'[C_i]$ has degree-deficit at least $r$. Thus there exists  $u_1,\ldots,u_r$ in $C_i$ missing edges. We add the edges $u_jw_j$ for $j \le r$. All the constraints are satisfied for $G'$ and $G'$ with degree-deficit $s-2$, which completes the proof.
\end{proof}

\begin{lemma}\label{lem:algo_poly}
 The proof of Theorem~\ref{thm:rea} can be turned into a polynomial time algorithm.
\end{lemma}

\begin{proof}[Sketch of the proof]
 In order to prove it let us prove that we can create a graph in $\cG(S,\mathcal{CC})$ with degree-deficit $s$ where $\ell(C) \le s \le u(C)$ (with the parities of $\ell(C)$ and $u(C)$) in polynomial time. For simplicity, in what follows, we will assume that $C=V$ (free to restrict the family $\mathcal{CC}$, we can assume that we are in that case). We prove by induction on the depth of the tree of the fragment that such a graph can be obtained in time $d \cdot P(n)$ where $d$ denotes the depth of the rooted tree. Note that when the rooted tree has depth $0$, the proof is immediate.

  First note the the proof of Lemma~\ref{lem:uc} is algorithmic and the algorithm runs in polynomial time. So we can construct a graph such that each set $C$ of $\mathcal{CC}$ induces a connected subgraph and with degree-deficit $u(V)$. Now we add edges in this graph. We stop either when the degree-deficit $s$ is reached (and we are done) or when no edge can be added further. Note that in this case, it only remains one vertex $v$ such that $s < d_i-d(v_i)$. Let $C_i$ be the child of $V$ containing $v_i$. As in the proof of Lemma~\ref{lem:decrease}, if there exists $C_j$ with $j \ne i$ such that the degree-deficit of $C_j$ is not $u(C_j)$, one can find an edge to delete in $G[C_j]$ to decrease the degree-deficit of $G[V]$. This edge can be found in polynomial time by Lemma~\ref{lem:GoodBadEdge}.
  Moreover, as in the proof of Lemma~\ref{lem:decrease}, we can assume that $\contracted{G}{V}$ is a star. Let $G'$ be the resulting graph.

  Let us denote by $s':= s + E_{G'}(C_i,V \setminus C)$. By induction, we can construct a graph on $C_i$ with degree-deficit $s'$ in time $(d-1) \cdot P(n)$. We complete this graph in such a way the graph induced by $V \setminus C_i$ is $G'[V \setminus C_i]$ and we add $E_{G'}(C_i,V \setminus C)$ edges between $C_i$ and its complement. The number of performed operations is indeed polynomial. So the complexity of this algorithm is indeed polynomial.
\end{proof}

\section{Connectivity of $\cG(S,\mathcal{CC})$ and approximation algorithm}\label{sec:connectivity}

\paragraph{Well-structured subtrees and extensions.}

\begin{figure}
 \centering
 \includegraphics[scale=0.6]{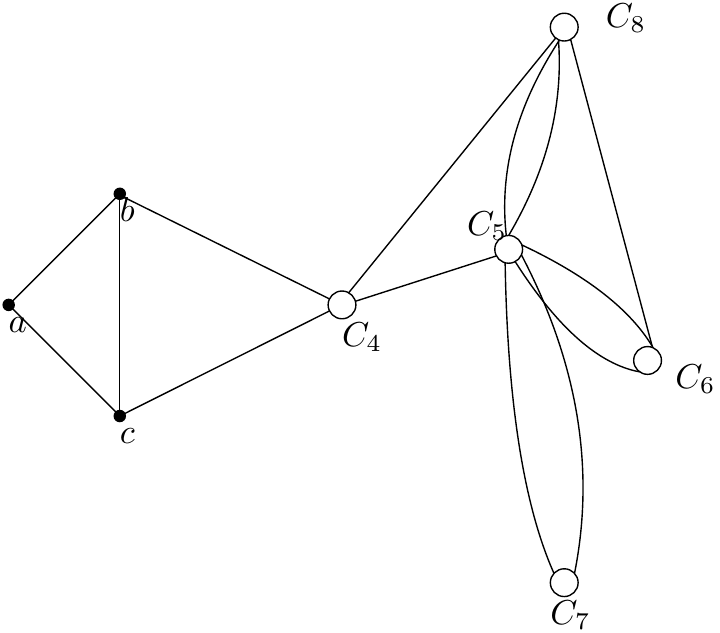}
 \caption{The graph $\contractedd{G}{T'}$ where the leaves of the well-structured subtree $T'$ are $a,b,c,C_4,C_5,C_6,C_7$. The graph $G$ is the one of Figure~\ref{fig:sfig2} and the tree $T$ is the one of Figure~\ref{fig:sfig4}.}
 \label{fig:contractedd}
\end{figure}

A subtree $T'$ of the tree of the fragment $T$ is \emph{well-structured} if
\begin{enumerate}[(i)]
 \item it contains the root, and
 \item if $u,v$ are two children of $w \in T'$ then either both $u$ and $v$ are in $T'$ or none of them is in $T'$.
\end{enumerate}
A set $C \in \mathcal{CC}$ is in a well-structured subtree $T'$ if the node labeled by $C$ is in $T'$. Given a well-structured subtree $T'$, we denote by $\mathcal{CC}_{T'}$ the subset of $\mathcal{CC}$ corresponding to the \emph{inner} nodes of $T'$, i.e. all the nodes of $T'$ but the leaves.
Note that the root of $T$ and the whole tree $T$ are well-structured.

The graph \emph{inherited from a well-structured subtree $T' \subseteq T$} denoted by $\contractedd{G}{T'}$ is the graph where the vertex set is the set of leaves of $T'$ and where there is an edge  between $X$ and $Y$ with multiplicity $\alpha$ if there are $\alpha$ edges with one endpoint in $X$ and one endpoint in $Y$ in $G$. Note that $\contractedd{G}{T}$ is the graph $G$.
Let $C \in \mathcal{CC}_{T'}$. We denote by $\contractedd{G}{T'}[C]$ the subgraph of $\contractedd{G}{T'}$ induced by the leaves of the subtree of $T'$ rooted at $C$.
A flip in $\contractedd{G}{T'}$ is \emph{$T'$-correct} if it maintains the connectivity of $\contractedd{G}{T'}[C]$ for any $C \in \mathcal{CC}_{T'}$.
Note that a correct flip is a $T$-correct flip.

For the sake of readability, vertices of $\contractedd{G}{T'}$ will be denoted with capital letters and vertices of the original graph $G$ will be denoted with lower case letters.

\begin{lemma}\label{lem:wellstructured1}
Let $S$ be a degree sequence and $\mathcal{CC}$ be a nested partition. Let $G$ and $H$ be two graphs in $\cG(S,\mathcal{CC})$. Let $T'$ be a well-structured subtree of the tree of the fragments $T$ and let $(AB,CD)\rightarrow(AC,BD)$ be a $T'$-correct flip.

We can find in polynomial time a flip $(ab,cd)\rightarrow(ac,bd)$ of $G$ where $a,b,c$ and $d$ are respectively in $G[A],G[B],G[C]$ and $G[D]$ which is correct for $G$.
Moreover if $AB$ and $CD$ are in $\contractedd{G}{T'} - \contractedd{H}{T'}$ then we can assume that $ab$ and $cd$ are in $G-H$.
\end{lemma}

\begin{proof}
Let $(AB,CD)\rightarrow(AC,BD)$ be a flip on $\contractedd{G}{T'}$. Then there exists a flip $(ab,cd)\rightarrow(ac,bd)$ on $G$ where $a,b,c$ and $d$ are respectively in $G[A],G[B],G[C]$ and $G[D]$. Indeed, by definition of $AB$ (resp. $CD$) in $\contractedd{G}{T'}$, it means that there is an edge between $a \in A$ and $b \in B$ (resp. $c \in C$ and $d \in D$). Since $A,B,C,D$ are disjoint, it is possible to replace $ab,cd$ by $ac,bd$ without creating any loop.  Moreover if $AB$ and $CD$ are in the symmetric difference of $\contractedd{G}$, then the number of edges between $A$ and $B$ is larger in $G$ than in $H$ and then there exists an edge $ab$ in the $G-H$. Let us denote by $G'$ the graph obtained after the flip $(ab,cd)\rightarrow(ac,bd)$ on $G$.

For every set $C' \in \mathcal{CC} \setminus \mathcal{CC}_{T'}$, the set $G'[C']$ is connected. Indeed, since $C' \notin  \mathcal{CC}_{T'}$, at most one of the four points $a,b,c,d$ are in $C'$. Thus the graph $\contracted{G'}{C'}=\contracted{G}{C'}$ is connected. Note that it holds for any $C' \notin  \mathcal{CC}_{T'}$, which is a family downward closed. So Lemma~\ref{lem:ContractEquiv} ensures that $G'[C']$ is connected since $\contracted{G}{C''}$ is connected for every $C'' \subseteq C'$.

Let us finally prove that the flip $(ab,cd)\rightarrow(ac,bd)$  maintains the connectivity of $G[C]$ for $C \in \mathcal{CC}_{T'}$. Assume by contradiction that $G[C]$ is not connected for some $C \in \mathcal{CC}_{T'}$. Amongst all such sets of $\mathcal{CC}_{T'}$, select $C$ minimal by inclusion. Let $u$ and $w$ be such that there is no path from $u$ to $w$ in $G[C]$. Let $U$ be the node of $\contractedd{G}{T'}$ containing $u$ and $W$ be the node of $\contractedd{G}{T'}$ containing $w$. Since the flip $(AB,CD)\rightarrow(AC,BD)$ is $T'$-correct, there is a path $U=U_0U_1,\ldots,U_\ell=W$ from $U$ to $W$ in $\contractedd{G'}{T'}[C]$. For every edge $U_iU_{i+1}$ of this path, there is an edge $u_iu_{i+1}'$ of $G$ with one endpoint in $U_i$ and one endpoint in $U_{i+1}$. Moreover, by minimality of $C$, for every pair $u_i,u_i'$ in $U_i$ there is a path in $G'[U_i]$ (and then in $G'[C]$) from $u_i$ to $u_i'$. Moreover since $U=U_0$ and $W=U_\ell$ are connected, there is a path from $u$ to $u_0$ in $G'[U]$ and a path from $u_\ell'$ to $w$ in $G'[W]$. So the path $P$ from $U$ to $W$ in $\contractedd{G'}{T'}[C]$ can be transformed into a path from $u$ to $w$ in $G[C]$, a contradiction with the fact that $G[C]$ is not connected.
\end{proof}

Lemma~\ref{lem:wellstructured1} permits to work with the graphs $\contractedd{G}{T'}$ and $\contractedd{H}{T'}$ and ensures that if we make a flip on one of these graphs, it can be simulated by a flip in the original graph. A natural question immediately arises, how can we ensure that the graph $G$ obtained after the flip still satisfies all the constraints? The next lemma will permit to answer this question.

Let $G$ and $H$ be two graphs in $\cG(S,\mathcal{CC})$. Let $T'$ be a well-structured subtree of the tree of the fragments $T$.
We can extend the notion of good and bad edges to the graphs $\contractedd{G}{T'}$ and $\contractedd{H}{T'}$. (It is \emph{good} if it is in both graphs and \emph{bad} otherwise).

Let $T$ be the tree of the fragments and $T_1$ be a well-structured subtree. The tree $T_2$ is an \emph{extension} of $T_1$ on \emph{extension node $C$} if $C$ is a leaf of $T_1$ and $T_2$ is $T_1$ plus all the children of $C$ in $T$. The set of children $\mathcal{X}$ of $C$ is then called the set of  \emph{special vertices of $\contractedd{G}{T_2}$}. Any flip between two edges of $\contractedd{G}{T_2}$ with at least one endpoint in $\mathcal{X}$ which:
 \begin{itemize}
  \item maintains the connectivity of $\contractedd{G}{T_2}[\mathcal{X}]$ and,
  \item does not create any edge in $\contractedd{G}{T_2} \setminus \mathcal{X}$.
 \end{itemize}
 is called a \emph{special flip}.



\begin{lemma}\label{lem:conditionflip}
 Let $S$ be a degree sequence and $\mathcal{CC}$ be a nested partition. Let $G$ in $\cG(S,\mathcal{CC})$. Let $T_2$ be an extension of a well-structured subtree $T_1$. Any special flip is $T_2$-correct.
\end{lemma}

\begin{proof}
Let us denote by $C$ the extension node and by $\mathcal{X}$ the set of special vertices. Let $T$ be the tree of the fragments.
Let $C' \in \mathcal{CC}_T$. Note that either $\mathcal{X}$ is included in the subtree rooted at $C'$ or $\mathcal{X}$ does not intersect it since $\mathcal{X}$ is the set of children of some node of the tree $T$. If $C'$ does not contain $\mathcal{X}$ in its subtree, then the result is immediate. Indeed, the flip does not affect any edge of $G[C']$ (since both edges of the flip have at least one endpoint in $\mathcal{X}$) and then the graph $\contractedd{G}{T_2}[C']$ is not modified and is still connected.

Assume now that $C'$ contains $\mathcal{X}$ in its subtree. Assume that there is no path from a vertex $U$ to $V$ in $\contractedd{G}{T_2}[C']$ after the flip. Let $P$ be a path from $U$ to $V$ in $\contractedd{G}{T_2}[C']$.
Since we only modify edges with at least one endpoint in $\mathcal{X}$, the path has to pass through $\mathcal{X}$. Let us prove that there still exists a path from $U$ (resp. $V$) to $\mathcal{X}$ after the flip. Since $\mathcal{X}$ still induces a connected subgraph after the flip by assumption, we will obtain a contradiction. Let $P'$ be a minimal path from $U$ to $\mathcal{X}$ in  $\contractedd{G}{T_2}$. Let $e$ be the last edge of $P'$ and $W$ the endpoint of $e$ not in $\mathcal{X}$. Since the path does not exist after the flip, the edge $e$ has been flipped. Since we do not create any edge in $\contractedd{G}{T_2} \setminus \mathcal{X}$, the edge $e$ has been replaced by an edge between $W$ and $W'$ where $W' \in \mathcal{X}$, a contradiction. So $\contractedd{G}{T_2}[C']$ is connected after the flip for any $C' \notin \mathcal{CC}_{T_2}$.
\end{proof}

Two graphs $G$ and $H$ \emph{agree on well-structured subtree $T'$} if $\contractedd{G}{T'} = \contractedd{H}{T'}$. Note that if $G$ and $H$ agree on $T$ then $G=H$.

The connectivity of $\cG(S,\mathcal{CC})$ and the approximation algorithm will follow from the next two lemmas. Before stating them formally, let us briefly describe the ideas of the proof. We will start with a graph the well-structured subtree $T'$ reduced to the root. This well-structured subtree will grow little by little during the proof until $T'=T$. Our goal consists in transforming $\contractedd{G}{T'}$ into $\contractedd{H}{T'}$ via special flips. Lemma~\ref{lem:conditionflip} ensures that this sequence of flips is $T'$-correct and
Lemma~\ref{lem:wellstructured1} ensures that this sequence can be adapted into a sequence of flips for $G$ that are correct.
However, in order to be able to find such a transformation we first need to transform the graphs in such a way the two graphs $\contractedd{G}{T'}$ and $\contractedd{H}{T'}$ have the same degree sequence. Lemma~\ref{lem:diffdegree} will ensure that it is possible to assume it. Then Lemma~\ref{lem:samedegree} will guarantee that a sequence of special flips permits to transform $\contractedd{G}{T'}$ into $\contractedd{H}{T'}$. So we finally obtain two graphs (still denoted by $G$ and $H$ for convenience) such that $\contractedd{G}{T'}=\contractedd{H}{T'}$. In that case, we perform an extension on a leaf of the subtree $T'$ and repeat the process until $T'=T$. At the end of this last step, we get $G=H$ since $G=\contractedd{G}{T'}=\contractedd{H}{T'}=H$.

\begin{lemma}\label{lem:diffdegree}
 Let $S$ be a degree sequence and $\mathcal{CC}$ be a nested partition. Let $T$ be the tree of the fragments. Let $T_2$ be an extension of a well-structured subtree $T_1$ on extension node $C$. Let $G,H$ be two graphs of $\cG(S,\mathcal{CC})$ that agree on $T_1$.

 We can find in polynomial time  a sequence of correct flips that transform $G$ into $G'$ and $H$ into $H'$ in such a way $\contractedd{G'}{T_2}$ and $\contractedd{H'}{T_2}$ have the same degree sequence and $G',H'$ still agree on $T_1$. The number of flips in the sequence is at most $\delta(\contractedd{G}{T_2},\contractedd{H}{T_2})/2$ and no flip modifies a good edge. Moreover we have $\delta(\contractedd{G'}{T_2},\contractedd{H'}{T_2}) \le \delta(\contractedd{G}{T_2},\contractedd{H}{T_2})$.
\end{lemma}
\begin{proof}
Let us denote by $V(T_2)$ the vertex set of both $\contractedd{G}{T_2}$ and $\contractedd{H}{T_2}$. Let $X$ be a vertex such that the degree of $X$ in $\contractedd{G}{T_2}$ is smaller than the degree of $X$ in $\contractedd{H}{T_2}$ (the case where the degree of $X$ is larger in $\contractedd{G}{T_2}$ is symmetric). Note that since $G$ and $H$ agree on $T_1$, $X$ must be a special vertex since the degree of a non special vertex is necessarily correct. In particular, there exists an edge $XY$ in $\contractedd{H}{T_2} - \contractedd{G}{T_2}$. We now distinguish two cases. \smallskip

\noindent\textbf{Case 1.} The degree of $Y$ in $\contractedd{G}{T_2}$ is smaller than the degree of $Y$ in $\contractedd{H}{T_2}$.  \\
Lemma~\ref{lem:GoodBadEdge} (applied to respectively $X$ and $Y$) ensures that there exists an edge $ab$ in the graph induced by $X$ and there exists an edge $cd$ in the graph induced by $Y$ such that their deletions do not disconnect any subgraph $G[C]$ for $C \in \mathcal{CC}$. We can perform the flip $(ab,cd)$ in the graph $G$. In the graph $\contractedd{G}{T_2}$, it corresponds to create twice the edge $XY$. Since $XY$ is in $\contractedd{H}{T_2} \Delta\contractedd{G}{T_2}$, the symmetric difference does not increase and since $Y$ is special, $\contractedd{G}{T_1}$ is not modified (the number of edges leaving the set of special vertices remain the same).
\smallskip

\noindent\textbf{Case 2.} The degree of $Y$ in $\contractedd{G}{T_2}$ is at least the degree of $Y$ in $\contractedd{H}{T_2}$.\\
Lemma~\ref{lem:GoodBadEdge} ensures that there exists an edge $ab$ in the graph induced by $X$ whose deletion does not disconnect any subgraph $G[C]$ for $C \in \mathcal{CC}$. First assume that $Y$ is a special vertex. Moreover, since the edge $XY$ is in $\contractedd{H}{T_2}$, there exists a vertex $Z \ne X$ incident to $Y$ in $\contractedd{G}{T_2}$. Let $yz$ be an edge of $G$ between these two sets. We perform the flip $(ab,yz) \rightarrow (ay,bz)$ in $G$. In $\contractedd{G}{T_2}$ is corresponds to create $XY$ and $XZ$ and to delete $YZ$. By definition of $XY$ is an edge of  does not increase the symmetric difference) and decreases the difference of degrees. Moreover $\contracted{G'}{C}$ is connected for every $C$ after the flip. Indeed, the deletion of $ab$ does not impact any set by Lemma~\ref{lem:GoodBadEdge}. Moreover, since both $X$ and $Y$ are special vertices, there always exist a path from $Z$ to $Y$ via special vertices and then the connectivity after the flip is still satisfied.

Assume now that $Y$ is not a special vertex. Since $G$ and $H$ agree on $\contractedd{H}{T_1}$, the number of edges from $Y$ to the extension node $C$ in both $G$ and $H$ is the same. And then the number of edges between $Y$ and special vertices is the same in $\contractedd{H}{T_2}$ and $\contractedd{G}{T_2}$. So there exists an edge $YZ$ where $Z$ is a special vertex. Let $yz$ be a corresponding edge of $G$. Then we perform the flip $(ab,yz) -> (ay,bz)$. As in the other case, one can prove that this flip is correct.
 \end{proof}

 \begin{lemma}\label{lem:samedegree}
  Let $S$ be a degree sequence and $\mathcal{CC}$ be a nested partition. Let $T$ be the tree of the fragments. Let $T_2$ be an extension of a well-structured subtree $T_1$ on extension node $C$. Let $G,H$ be two graphs of $\cG(S,\mathcal{CC})$ that agree on $T_1$ and such that $\contractedd{G}{T_2}$ and $\contractedd{H}{T_2}$ have the same degree sequence.

  We can find in polynomial time a sequence of at most $\delta(\contractedd{G}{T_2},\contractedd{H}{T_2})$ special flips transforming $\contractedd{G}{T_2}$ into $\contractedd{H}{T_2}$ which only flips bad edges.
 \end{lemma}

 Note that Lemma~\ref{lem:conditionflip} then ensures that this sequence of special flips is a sequence of correct flips.

\begin{proof}
Let $ \mathcal{X}$ be the set of special vertices of $T_2$. Let us denote by $\delta(\mathcal{X})$ the number of edges between $\mathcal{X}$ and its complement in $\contractedd{G}{T_2}$. Note that $\delta(\mathcal{X})$ is the same for $G$ and $H$ since they agree on $T_1$. Let us first create the following graphs $G'$ and $H'$ defined on the same set of vertices with the same degree sequence.

The vertex set of these graph is $\mathcal{X}$ plus $\delta_G(\mathcal{X})$ vertices.
Restricted to the set $\mathcal{X}$, the graphs $G'$ and $H'$ respectively induce the graphs $\contracted{G}{C}$ and $\contracted{H}{C}$.


Let us denote by $\ell$ the number of edges between $\mathcal{X}$ and its complement in both $G$ and $H$. Note that this value is the same since by assumption $G$ and $H$ agree on $T_2$ and $\mathcal{X}$ is the set of special vertices. Now we create $\ell$ new vertices which will be of degree one in both $G'$ and $H'$. Let $W$ in $V(\contractedd{G}{T_2}) \setminus \mathcal{X}$. Since $G$ and $H$ agree on $T_2$, the number of edges between $W$ and $\mathcal{X}$ is the same in both $\contractedd{G}{T_2}$ and $\contractedd{H}{T_2}$. Let $\ell'$ this value. We create $\ell'$ vertices $W_1,\ldots,W_{\ell'}$ of degree one in $G'$ and $H'$. These vertices are connected to respectively neighbors of $W$ in $\contractedd{G}{T_2}$ and $\contractedd{H}{T_2}$.

We claim that $G'$ and $H'$ have the same degree sequence. Indeed, since there are as many edges incident to $W$ in both $\contractedd{G}{T_2}$ and $\contractedd{H}{T_2}$ for every $W$, the number of pending edges is the same in both graphs. Since we have added to $X$ in $\mathcal{X}$ exactly as many pending edges as neighbors of $X$ in $V(\contractedd{G}{T_2}) \setminus X$, the vertices of $\mathcal{X}$ in $G'$ and $H'$ are note modified. Since they were initially the same in $\contractedd{G}{T_2}$ and $\contractedd{H}{T_2}$, they are still the same in $G'$ and $H'$. So $G'$ and $H'$ are two graphs with the same degree sequence.

Note that both graphs are connected and also remark that, since all the vertices of $G'$ and $H'$ but the ones of $\mathcal{X}$ have degree one, any connected graph $K$ with the same degree sequence satisfies that $K[\mathcal{X}]$ is connected.

By Theorem~\ref{thm:connected}, it is possible to transform $G'$ into $H'$ via a sequence of at most $2\delta(\contractedd{G}{T_2},\contractedd{H}{T_2})$ flips maintaining the connectivity of the graph. Since pending vertices are associated with vertices of $\contractedd{G}{T_2} \setminus \mathcal{X}$ (and $\contractedd{H}{T_2} \setminus \mathcal{X}$), this sequence of flips can indeed by transformed into a sequence of flips in $\contractedd{G}{T_2}$.
This sequence of flips can be adapted  transforms $\contractedd{G}{T_2}$ and $\contractedd{H}{T_2}$. Since the intermediate graphs remains connected if and only if $\mathcal{X}$ induces a connected subgraph and no two vertices of degree one are linked by an edge, the corresponding sequence of flips satisfies that, for any intermediate graph $D$, $\contractedd{D}{T_2}[\mathcal{X}]$ is connected. Moreover all the edges flipped have at least one endpoint in $\mathcal{X}$ and no edge is created in $V(\contractedd{D}{T_2}) \setminus \mathcal{X}$. Thus all the flips of the sequence are special flips. Thus by Lemma~\ref{lem:conditionflip}, the sequence of flips is $T_2$-correct.

Finally, Theorem~\ref{thm:connected} ensures that there exists a transformation from $\contractedd{G}{T_2}$ to $\contractedd{H}{T_2}$ that only flips edges in the symmetric difference. This plus Lemma~\ref{lem:wellstructured1} guarantees that there exists a transformation that never modify good edges during the transformation from $\contractedd{G}{T_2}$ to $\contractedd{H}{T_2}$, which completes the proof of the lemma.
\end{proof}

Let us make the following remark:
\begin{rem}
The algorithms proposed in Lemmas~\ref{lem:samedegree} and~\ref{lem:diffdegree} never create an edge with both endpoint in $V(\contractedd{G}{T_2})\ S$ where $S$ is the set of special vertices.
\end{rem}
 For Lemma~\ref{lem:samedegree} the proof is immediate since flips are special. The choice of edges to flip in Lemma~\ref{lem:diffdegree} ensures that the conclusion holds. Indeed it holds by construction in Case 2 and holds in Case 1 since $Y$ must be a special vertex (since the degree of a vertex is the same in $\contractedd{G}{T_2}$ and $\contractedd{H}{T_2}$ for vertices which are not special).

\paragraph{The algorithm.}
Let us now present the algorithm to compute a sequence of correct flips that transform $G$ into $H$:

 \begin{algorithm}[h!]
\caption{Find a sequence of flip that transforms $G$ into $H$}
\label{alg:calculw}
\begin{algorithmic}[1]
  \STATE Compute the tree of the fragments $T$ rooted at $r$.
  \STATE $T'\leftarrow $ Root of $T$.
  \STATE $G_1 \leftarrow G, H_1 \leftarrow H$.
  \WHILE{$T' \neq T$}
  \STATE Let $C$ be a leaf of $T'$ which is an internal node of $T$.
  \STATE Let $T"$ be well-structured subtree inherited from $T'$ on extension node $C$.
  \STATE\label{algo1} Transform $G_1$ and $H_1$ into $G_2$ and $H_2$ via a sequence of $T"$-correct flips in such a way they have the same degree sequence on $V(\contractedd{G}{T"})$ using Lemma~\ref{lem:diffdegree}.
  \STATE\label{algo2} Transform $G_2$ into $G_3$ via a sequence of $T"$-correct flips in such a way the two graphs $G_3$ and $H_2$ agree on $T"$ using Lemma~\ref{lem:samedegree}.
  \STATE $G_1 \leftarrow G_3, H_1 \leftarrow H_2, T' \leftarrow T"$.
  \ENDWHILE
\end{algorithmic}
\label{alg1}
\end{algorithm}

We have all the ingredients that guarantee that the algorithm works properly. Indeed, Lemma~\ref{lem:diffdegree} ensures that it is possible to transform the graph in such a way all the connectivity constraints are still satisfied and the two graphs have the same degree sequence on $T"$. Moreover this sequence can be found in polynomial time. So the step of line~\ref{algo1} can be performed in polynomial time. Moreover, Lemma~\ref{lem:wellstructured1} ensures that the transformation using $T"$-correct flips that transform $G_2$ into $G_3$ also is correct in the whole graph. And Lemma~\ref{lem:samedegree} ensures that this sequence exists. So the step of line~\ref{algo2} can be performed (in polynomial time).
When the algorithm stops, the graph $G_1$ and $H_1$ agree on $T$ and then the two graphs are the same.

\begin{theorem}
  Algorithm~\ref{alg1} provides a sequence of flips of length at most $(2d+1) \delta(G,H)$ transforming $G$ into $H$ in $\cG(S,\mathcal{CC})$. In particular, it provides a $(8d+4)$-approximation algorithm of the shortest sequence.
 \end{theorem}
\begin{proof}
Let $T_2$ be an extension of $T_1$ and $G,H$ be two graphs $T_1$-correct. At each step of the algorithm, we need $\delta(\contractedd{G}{T_2},\contractedd{H}{T_2})/2$ flips to equilibrate the degrees by Lemma~\ref{lem:diffdegree} and $\delta(\contractedd{G}{T_2},\contractedd{H}{T_2})$ flips to transform $\contractedd{G}{T_2}$ into $\contractedd{H}{T_2}$ by Lemma~\ref{lem:samedegree}. So in total to transform a graph that is $T_1$-correct into a graph that is $T_2$-correct, one needs $3/2 \delta(\contractedd{G}{T_2},\contractedd{H}{T_2})$ flips.

Let us now count how many times an edge of the symmetric difference can appear in the symmetric difference of $\delta(\contractedd{G}{T'},\contractedd{H}{T'})$ for some well-defined subtree $T'$. In the proofs of Lemma~\ref{lem:diffdegree} and Lemma~\ref{lem:samedegree}, we only modify edges incident to at least one special vertex. So if an edge $e$ is modified, then one of its endpoint is a special vertex. Moreover, at the end of the algorithm given by Lemma~\ref{lem:samedegree}, the edge on which $e$ is flipped is an edge with at least one endpoint in the set of special vertices. So later on the algorithm, it can only be modified again when one of its endpoint become a special vertex. Since the depth of the tree $T$ is at most $d$, an edge can have an endpoint on a special vertex only at most $2d$ times. So in total any edge of $G\Delta H$ appears in at most $2d$ graphs $\contractedd{G}{T'}\Delta\contractedd{H}{T'}$ for some well-defined subtree $T'$.

Finally, there exists a transformation from $G$ to $H$ in $\cG(S,\mathcal{CC})$ of length at most $3d \delta(G,H)$. \smallskip

This bound can be slightly improved. Indeed one can notice that if some edge $e$ is used in the proof of Lemma~\ref{lem:diffdegree}, then the edge becomes an edge of $\contractedd{G}{T'}$ for the current tree $T'$. And this edge will never become an edge of $G[C]$ for $C \notin \mathcal{CC}_{T''}$ for some future extension $T''$ of $T'$. So any edge of the symmetric difference can actually be considered at most once in this lemma. So in total, the number of edges modified by the application of Lemma~\ref{lem:diffdegree} for \textit{all} the well-defined tree $T'$ is at most $\delta(G,H)$. Using this fact, we obtain a transformation of length at most  $(2d+1) \delta(G,H)$.

Since an optimal solution needs at least $\delta(G,H)/4$ flips, our algorithm provides a $(8d+4)$-approximation algorithm.
\end{proof}

\section{Conclusion and open problems}
In Section~\ref{sec:approx-classique}, we provide a $4$-approximation algorithm to transform a connected multigraph into another. In this paper, we were simply interested in the existence of a constant approximation algorithm in order to obtain an approximation algorithm for the generalized problem. We did not make any attempts to optimize our bound. It is likely that a more careful analysis provides a better approximation ratio. In particular, if one can prove that there always exists a flip that creates a good edge (without breaking ones) that does not disconnect the graph, then we would immediately obtain a $2$-approximation algorithm.
Recently, Bereg and Ito~\cite{Bereg17} provide a $3/2$-approximation algorithm to transform a multigraph into another based on the symmetric circuit partition, beating the trivial $2$-approximation algorithm in that case. A similar technique might be useful for significantly improve the approximation ratio for connected graphs.


The approximation ratio of Theorem~\ref{thm:main} depends on the depth of the tree of the fragment. Can we avoid this dependency and simply find an approximation algorithm that does not depend of the depth? Note that we did not try to optimize the constant in front of the linear function of $d$ in the approximation ratio that is probably easily improvable.

\begin{figure}
 \centering
 \includegraphics[scale=.6]{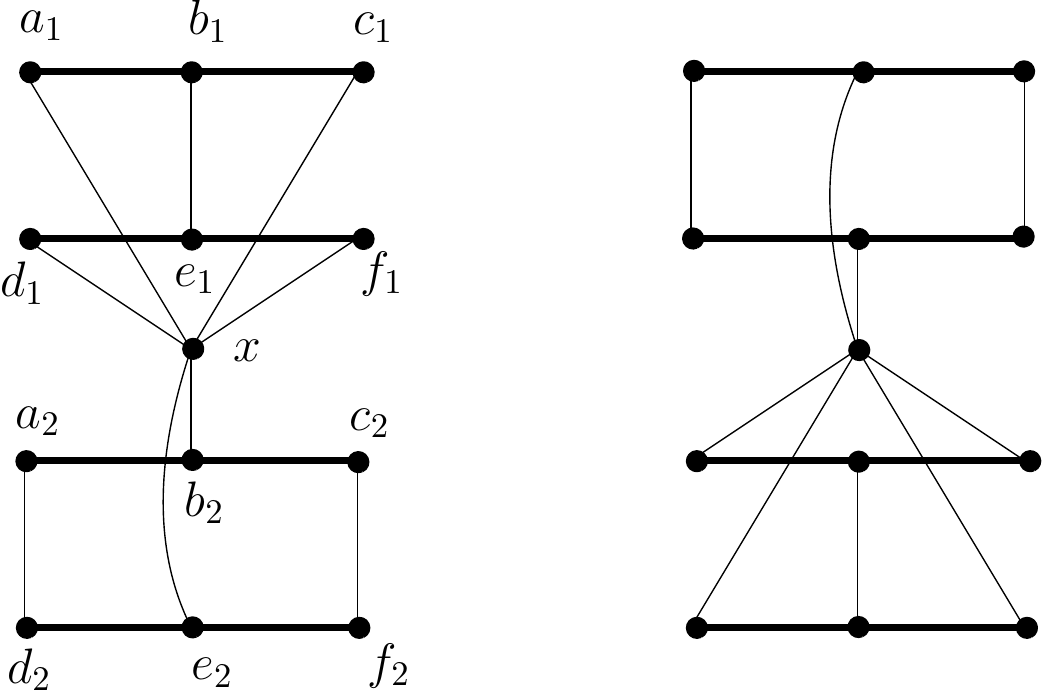}
 \caption{It is impossible to transform the left graph into the right graph when the set of subsets that have to induce connected subgraphs is the set of thick edges plus the sets  $\{a_1,b_1,d_1,x\}, \{a_1,b_1,e_1,x\}, \{b_1,c_1,f_1,x \},$ $\{b_1,c_1,e_1,x \}, \{a_2,b_2,d_2,x\}, \{a_2,b_2,e_2,x\},$ $\{b_2,c_2,f_2,x \}$ and $\{b_2,c_2,e_2,x \}$.
 }
 \label{fig:counter}
\end{figure}

In practice, instead of one tree of the fragments, we are often given a collection of trees of the fragments instead of one. It means  that subsets of vertices that have to be connected might intersect and not be contained one in the other. One can wonder if it is still true that the reconfiguration graph is connected in this setting? Unfortunately the answer to this question is negative, for instance in the graph provided in Figure~\ref{fig:counter}.

If yes, does it always exist a transformation that is linear in the size of the symmetric difference and can we approximate it in polynomial time?

\paragraph{Acknowledgments.} The authors want to thank the anonymous reviewers of WAOA for their careful reading of the paper which permits to significantly improve its quality.

\bibliographystyle{abbrv}
\bibliography{bib.bib}

\clearpage

\appendix

\section{NP-Hardness}\label{sec:NP-Hardness}


In this section we show that computing the distance
$dist_{\cG(S,\mathscr{C})}(G,H)$ or $dist_{\cG(S)}(G,H)$
between two graphs $G$ and $H$ is NP-complete where $dist_{\cG(S,\mathscr{C})}(G,H)$ and $dist_{\cG(S)}(G,H)$ correspond to the shortest sequence of flips needed to transform $G$ into $H$ in respectively $\cG(S,\mathscr{C})$ and $\cG(S)$. In \cite{Will99}, Will
proved that computing the distance
$dist_{\cG(S,\mathscr{S})}(G,H)$ between two simple graphs is
NP-complete. To prove it , he shows that the distance depends on the maximum size of
a \emph{symmetric circuit partition} of $G \Delta H$ and proves that computing the maximum size of symmetric circuit partition is NP-complete.

A \emph{circuit} is a close walk which passes through any edge at most
once. A  \emph{symmetric circuit} in $G \Delta H$ is a circuit whose edges
alternate between $G − H$ and $H − G$. A \emph{symmetric circuit partition} of $G \Delta H$ is a set of pairwise edge-disjoint
symmetric circuits using all the edges of $G \Delta H$. We denote by $m(G,H)$
the maximum size of symmetric circuit partition of $G \Delta H$.

\begin{theorem}[\cite{Will99}]\label{theorem:will}
    Any pair of simple graphs $G$ and $H$ realizing the same degree
    sequence satisfies
    \[
        dist_{\cG(S,\mathscr{S})}(G,H)=\frac{\delta(G,H)}{2}
        - m(G,H)
    \]
\end{theorem}
Let us denote $\psi(G,H):=\frac{\delta(G,H)}{2}
- m(G,H)$.
It was already noticed in \cite{Majcher87} that
$dist_{\cG(S,\mathscr{S})}(G,H)\leq \psi(G,H)$. Theorem \ref{theorem:will} was rediscovered in
\cite{Bereg17} and generalized in \cite{erdos13}.

A simple rewriting of the proof in \cite{Will99} described below permits to prove the following:

\begin{theorem}\label{theorem:distmulti}
    Any pair of loop-free (multi)graphs $G$ and $H$ realizing the same degree
    sequence satisfies
    \[
        dist_{\cG(S)}(G,H)=\psi(G,H)    \]
\end{theorem}

To prove that  $dist_{\cG(S,\mathscr{S})}(G,H)\leq \psi(G,H)$ Will exhibits a sequence of flip of size $\frac{\delta(G,H)}{2}
        - k$ given a symmetric circuit partition of $G∆H$ of size $k$.
        Even if it was not originally designed for multigraphs, the sequence is also
        valid in the case where $G$ or $H$ have  multiple edges.
        Hence we have $dist_{\cG(S)}(G,H)\leq \psi(G,H)$

To get the other direction, in \cite{Will99,Bereg17} the authors proved the following.
\begin{lemma}[\cite{Will99,Bereg17}]
    Let $G'$ be a graph obtained by a flip from $G$. Then
    \[ \psi(G',H)\geq \psi(G,H)-1\]
\end{lemma}

We remark that the translation of the proof in \cite{Bereg17} for multigraphs is immediate, which proves that $dist_{\cG(S)}(G,H)\leq \psi(G,H)$ and thus proves Theorem \ref{theorem:distmulti}.

Since the computing $m(G,H)$ is NP-complete, Theorem \ref{theorem:distmulti} ensures that computing $dist_{\cG(S)}(G,H)$ is also NP-complete.  Note that the inclusion in NP is immediate since
there always exist transformation of linear size, and then a polynomial size certificate.

\begin{corollary}\label{cor:dist-multi-NP-hard}
    Given two loop-free (multi)graphs $G$ and $H$, computing
    $dist_{\cG(S)}(G,H)$ is NP-complete.
\end{corollary}

Furthermore the restriction of connected graphs does not change the complexity.

\begin{corollary}\label{cor:dist-multi-connected-NP-hard}
    Given two connected loop-free (multi)graphs $G$ and $H$, computing
    $dist_{\cG(S,\mathscr{C})}(G,H)$ is NP-complete.
\end{corollary}
\begin{proof}
  We show that the computation of $dist_{\cG(S)}$ can be reduced to the computation of $dist_{\cG(S,\mathscr{C})}$. Let $G$ and $H$ be two graphs and let $G'$ and $H'$ be the graphs obtained from $G$ and $H$ by adding a universal vertex $u$. We show that $dist_{\cG(S)}(G,H) = dist_{\cG(S,\mathscr{C})}(G',H')$. Indeed, notice that we have $G \Delta H=G' \Delta H'$. So by Theorem~\ref{theorem:distmulti}, $dist_{\cG(S)}(G,H)=dist_{\cG(S)}(G',H')$. First since $\cG(S,\mathscr{C})$ is a subgraph of $\cG(S,\mathscr{C})$, $dist_{\cG(S)}(G',H') \leq dist_{\cG(S,\mathscr{C})}(G',H')$. Now consider the shortest path between $G'$ and $H'$ in $\cG(S)$. Since the shortest sequence of flips between $G'$ and $H'$ in $\cG(S)$ only modify edges of $G'\Delta H'$, every intermediate graph has $u$ as universal vertex. Thus they are all connected, which means that the shortest path in $\cG(S)$ is also a path in $\cG(S,\mathscr{C})$. So $dist_{\cG(S)}(G',H')\geq dist_{\cG(S)}(G',H')$, which concludes the proof.
\end{proof}

Since computing the distance in $\cG(S,\mathcal{CC})$ is a generalization of connected loop-free multigraphs, computing the distance in $\cG(S,\mathcal{CC})$ also is NP-hard. The belonging to NP is immediate since we proved in Section~\ref{sec:connectivity} that the distance between any pair of graphs is polynomial.

\end{document}